\documentclass[sigconf,authorversion,nonacm]{acmart}
\usepackage{xcolor}
\usepackage{enumerate}
\usepackage{thm-restate}
\usepackage[ruled]{algorithm2e}

\theoremstyle{acmplain}
\newtheorem{thrm}{Theorem}
\newtheorem{lem}[thrm]{Lemma}
\newtheorem{prop}[thrm]{Proposition}

\newcommand{\N}{\mathbb{N}}
\newcommand{\Q}{\mathbb{Q}}

\newcommand{\Rp}{\mathbb{R}_{\geq 0}}

\newcommand{\A}{\mathbb{A}}
\newcommand{\B}{\mathbb{B}}

\newcommand{\R}{\mathbb{R}}

\newcommand{\mC}{\mathcal{C}}
\newcommand{\mH}{\mathcal{H}}
\newcommand{\mM}{\mathcal{M}}
\newcommand{\mR}{\mathcal{R}}
\newcommand{\mF}{\mathcal{F}}

\newcommand{\mS}{\mathcal{S}}

\newcommand{\bv}{\boldsymbol{v}}
\newcommand{\bx}{\boldsymbol{x}}
\newcommand{\ba}{\boldsymbol{a}}
\newcommand{\bc}{\boldsymbol{c}}
\newcommand{\bh}{\boldsymbol{h}}
\newcommand{\bg}{\boldsymbol{g}}

\newcommand{\bw}{\boldsymbol{w}}
\newcommand{\be}{\boldsymbol{e}}

\newcommand{\bff}{\boldsymbol{f}}

\newcommand{\bzer}{\boldsymbol{0}}

\newcommand{\card}{\operatorname{card}}
\newcommand{\cone}{\boldsymbol{cone}}
\newcommand{\lin}{\boldsymbol{lin}}

\begin{document}

\title{Termination of linear loops under commutative updates}

\author{Ruiwen Dong}
\email{ruiwen.dong@kellogg.ox.ac.uk}
\affiliation{%
  \institution{University of Oxford}
  \streetaddress{Department of Computer Science, 7 Parks Road}
  \city{Oxford}
  \country{United Kingdom}
  \postcode{OX1 3QG}
}

\begin{abstract}
We consider the following problem: given $d \times d$ rational matrices $A_1, \ldots, A_k$ and a polyhedral cone $\mathcal{C} \subset \mathbb{R}^d$, decide whether there exists a non-zero vector whose orbit under multiplication by $A_1, \ldots, A_k$ is contained in $\mathcal{C}$.
This problem can be interpreted as verifying the termination of multi-path while loops with linear updates and linear guard conditions.
We show that this problem is decidable for \emph{commuting} invertible matrices $A_1, \ldots, A_k$.
The key to our decision procedure is to reinterpret this problem in a purely algebraic manner.
Namely, we discover its connection with modules over the polynomial ring $\mathbb{R}[X_1, \ldots, X_k]$ as well as the polynomial semiring $\mathbb{R}_{\geq 0}[X_1, \ldots, X_k]$.
The loop termination problem is then reduced to deciding whether a submodule of $\left(\mathbb{R}[X_1, \ldots, X_k]\right)^n$ contains a ``positive'' element.
\end{abstract}



\maketitle

\begin{acks}
The author acknowledges support from UKRI Frontier Research Grant EP/X033813/1.
\end{acks}

\section{Introduction}\label{sec:intro}
We consider the termination problem for multipath (or branching) linear loops over the real numbers.
We are concerned with problems of the following form: given $d \times d$ matrices $A_1, \ldots, A_k$ and a subset $\mC$ of $\R^d$, decide whether there exists a vector $\bv \in \mC$ whose orbit under multiplication by $A_1, \ldots, A_k$ is contained in $\mC$.
In other words, denote by $\langle A_1, \ldots, A_k \rangle$ the monoid generated by $A_1, \ldots, A_k$, we want to decide whether there exists $\bv \in \mC$ such that $A \bv \in \mC$ for all $A \in \langle A_1, \ldots, A_k \rangle$.
Such a problem can be interpreted as the non-termination of the following linear loop:

\begin{equation}\label{eq:multloop}
\textit{while } \bx \in \mC \textit{ do } (\bx \coloneqq A_1 \bx \textit{ or } \cdots \textit{ or } \bx \coloneqq A_k \bx).
\end{equation}
The problem is whether there exists an initial value $\bv$ such that the loop never terminates.
Analysis of loops of this form is an important part of analysing more complex programs~\cite{jeannet2014abstract,kincaid2019closed}.
Various tools such as computing polynomial invariants~\cite{hrushovski2018polynomial}, porous invariants~\cite{lefaucheux2021porous}, ranking functions~\cite{ben2014ranking, ben2019multiphase} have been developed to analyse properties of such multipath loops.

Termination of linear loops with a \emph{single} update ($k = 1$) has been intensely studied for the last twenty years.
Tiwari~\cite{tiwari2004termination} first showed its decidability in the case where the guard condition $\mC$ is a convex polyhedron.
Subsequently, Braverman~\cite{braverman2006termination} extended the decidability result to the case where $\mC$ is the set of rational points in a convex polyhedron; and Hosseini, Ouaknine, Worrell~\cite{HosseiniO019} further extended the result to integer points.\footnote{The original results all concern single \emph{affine} updates. However, they are equivalent to the case of \emph{linear} updates by increasing the dimension by one and homogenizing the coordinates.}

In the above results, the decidability relies on the freedom of choosing the initial value $\bv$ to circumvent the difficulties inherent in deciding termination on a single initial value.
In fact, when $\mC$ is a closed halfspace, termination of single-path linear loops for a \emph{fixed} initial value, that is,
\begin{equation}
\bx \coloneqq \bv; \textit{ while } \bx \in \mC \textit{ do } (\bx \coloneqq A \bx),
\end{equation}
subsumes the \emph{Positivity Problem}~\cite{ouaknine2014positivity}.
The Positivity Problem at dimension greater than five has been shown to be intrinsically hard; and its decidability would entail major breakthroughs in open problems in Diophantine approximation~\cite{ouaknine2014positivity}.
Therefore, the freedom of choice for the initial value is crucial in proving decidability results.

In this paper we consider termination of the multipath loop~\eqref{eq:multloop} under \emph{commutative} updates, that is, when $A_1, \ldots, A_k$ are commuting matrices.
In~\cite{babai1996multiplicative}, Babai, Beals, Cai, Ivanyos and Luks proved various results on reachability problems of multipath linear loops under commutative updates.
A continuous analogue of the reachability problem has been studied in~\cite{ouaknine2019decidability}.
For the termination problem, we restrict to the case where $\mC$ is a \emph{polyhedral cone}, that is, 
\[
\mC \coloneqq \left\{\bx \in \R^d \;\middle|\; \bc_i^{\top} \bx \geq 0, i = 1, \ldots, n \right\}
\]
for some $\bc_1, \ldots, \bc_n \in \R^d$.
We will also exclude the initial value $(0, \ldots, 0)$ since it is trivially invariant under linear updates.
For the sake of simplicity, we will suppose $A_1, \ldots, A_k$ to have rational entries and are invertible, and that $\bc_1, \ldots, \bc_n \in \Q^d$.
Our main result (see Theorem~\ref{thm:main}) is that termination of such multipath linear loops is decidable.

If one tries to adapt the proofs for the single update case to multiple updates, natural attempts will start by simultaneously diagonalizing the update matrices and analysing their actions on the cone geometrically.
Unfortunately such analysis depends on classifying the complex eigenvalues of all the update matrices.
Besides the unavoidable tedious case analysis, it is also not clear how to circumvent certain problems related to Diophantine approximation, as such circumvention appears ad hoc in the single update case.
Our main contribution is a completely algebraic approach that overcomes this difficulty.
In particular, using the duality of linear programming, we reduce the problem to deciding whether an orbit cone is \emph{salient}.
We show that deciding whether an orbit is salient is equivalent to deciding whether a submodule of $\left(\mathbb{R}[X_1, \ldots, X_k]\right)^n$ contains an element with only positive coefficients.
We then use a result of Einsiedler, Mouat, and Tuncel~\cite{einsiedler2003does} and the decidability of first order theory of the reals to deduce the desired result.

\section{Preliminaries}\label{sec:prelim}
\subsection{Convex geometry}
A \emph{convex cone} in $\R^d$ is a subset $\mC \subset \R^d$ that satisfies $r \in \R_{\geq 0}, \bv \in \mC \implies r \cdot \bv \in \mC$ and $\bv, \bw \in \mC \implies \bv + \bw \in \mC$.
In this article, we will use the word \emph{cone} to imply convex cones.
A \emph{polyhedral} cone is a cone of the form
\[
\mC \coloneqq \left\{\bx \in \R^d \;\middle|\; \bc_i^{\top} \bx \geq 0, i = 1, \ldots, n \right\}.
\]
for finitely many given non-zero vectors $\bc_1, \ldots, \bc_n \in \R^d$.

Given mutually commuting matrices $A_1, \ldots, A_k \in \R^{d \times d}$, and a vector $\bv \in \R^d$, denote by
\[
\langle A_1, \ldots, A_k \rangle \cdot \bv \coloneqq \left\{ A_1^{n_1} \cdots A_k^{n_k} \bv \;\middle|\; n_1, \ldots, n_k \in \N\right\}
\]
the orbit of $\bv$ under multiplication by $A_1, \ldots, A_k$.
Similarly, given a cone $\mC \subset \R^d$, denote by
\[
\langle A_1, \ldots, A_k \rangle \cdot \mC \coloneqq \left\{ A_1^{n_1} \cdots A_k^{n_k} \bv \;\middle|\; n_1, \ldots, n_k \in \N, \bv \in \mC \right\}
\]
the orbit of $\mC$ under multiplication by $A_1, \ldots, A_k$.

Given an arbitrary subset $\mS \subset \R^d$, denote by $\cone(\mS)$ the smallest cone in $\R^d$ containing $\mS$.
Similarly, denote by $\lin(\mS)$ the smallest linear space in $\R^d$ containing $\mS$.

A \emph{closed halfspace} of a $\R^d$ is a subset $\mH$ defined by $\mH \coloneqq \{\bx \mid \bv^{\top} \bx \geq 0\}$ for some $\bv \in \R^d \setminus \{0^d\}$.
Here, $0^d$ denotes the zero vector of $\R^d$.
By the Supporting Hyperplane Theorem~\cite[Chapter~2.5.2]{boyd2004convex} (applied at the point $0^d$), every cone $\mC \subset \R^d$ is either contained in some closed halfspace or it is $\R^d$ itself.

A cone $\mC$ is called \emph{salient} if $\mC \cap - \mC = \{0^d\}$, where $- \mC \coloneqq \{- \bv \mid \bv \in \mC\}$.
An arbitrary set $\mS \subset \R^d$ is called salient if $\cone(\mS)$ is salient.
Obviously every salient cone (and hence every salient set) is contained in a closed halfspace because a salient cone cannot be $\R^d$.

\subsection{Modules over polynomial rings}
Fix an integer $k \geq 1$.
In this paper we consider the polynomial ring $\A \coloneqq \R[X_1, \ldots, X_k]$ and its sub-semiring
\begin{multline*}
    \A^+ = \R_{\geq 0}[X_1, \ldots, X_k] \\
    \coloneqq \Bigg\{\sum_{n_1, \ldots, n_k \in \N} r_{n_1, \ldots, n_k} X_1^{n_1} \cdots X_k^{n_k} \;\Bigg|\; r_{n_1, \ldots, n_k} \in \R_{\geq 0}, \\
    r_{n_1, \ldots, n_k} = 0 \text{ for all but finitely many } r_{n_1, \ldots, n_k} \Bigg\}.
\end{multline*}
That is, $\A^+ \subset \A$ is the sub-semiring of polynomials with non-negative coefficients.
Define additionally $\A^{++} \coloneqq \A^+ \setminus \{0\}$.

Let $\A^n \coloneqq \{(f_1, \ldots, f_n) \mid f_1, \ldots, f_n \in \A\}$.
An $\A$-\emph{submodule} of $\A^n$ is a subset $\mM \subset \A^n$ such that $\bv, \bw \in \mM \implies \bv + \bw \in \mM$ and $f \in \A, \bv \in \mM \implies f \cdot \bv \in \mM$.
Here $f \cdot (f_1, \ldots, f_n) \coloneqq (f f_1, \ldots, f f_n)$.

For example, given $\bg_1, \ldots, \bg_m \in \A^n$, the set
\begin{multline*}
    \mM = \bg_1 \A + \cdots + \bg_m \A \\
    \coloneqq \{f_1 \cdot \bg_1 + \cdots + f_m \cdot \bg_m \mid f_1, \ldots, f_m \in \A\}
\end{multline*}
is an $\A$-submodule of $\A^n$.
In this case we will say that $\bg_1, \ldots, \bg_m$ form a \emph{basis} of $\mM$. 
The above definition still works when one replaces $\A$ by any commutative ring.

Naturally, we define the following subset of $\A^n$:
\[
\left(\A^+\right)^n \coloneqq \{(f_1, \ldots, f_n) \mid f_1, \ldots, f_n \in \A^+\}.
\]
The key ingredient of our paper is a deep result by Einsiedler, Mouat, and Tuncel~\cite{einsiedler2003does} that characterizes when an $\A$-submodule of $\A^n$ contains an element of $\left(\A^+\right)^n \setminus \{0^n\}$.
We will state this result in Section~\ref{sec:pos}.

\section{Overview of the decision procedure}\label{sec:overview}
In this section we give an overview of the procedure that decides (non-)termination of the multipath loops under commutative updates.
In rigorous terms, given commuting matrices $A_1, \ldots, A_k$ and a polyhedral cone $\mC$, we want to decide whether there exists a vector $\bv \in \R^d \setminus \{0^d\}$, such that $\langle A_1, \ldots, A_k \rangle \cdot \bv \subset \mC$.

Let the cone $\mC$ be given by
\begin{equation}\label{eq:defC}
    \mC \coloneqq \left\{\bx \in \R^d \;\middle|\; \bc_i^{\top} \bx \geq 0, i = 1, \ldots, n \right\}.
\end{equation}
We define its \emph{dual cone} $\mC^*$ as
\begin{equation}\label{eq:defCstar}
    \mC^* \coloneqq \left\{\bx \in \R^d \;\middle|\; \bx = r_1 \bc_1 + \cdots + r_n \bc_n, r_1, \ldots, r_n \in \Rp \right\}.
\end{equation}

Our first lemma is a duality type argument that reduces the termination problem to deciding whether an orbit is contained in some halfspace.

\begin{restatable}{lem}{lemdual}\label{lem:dual}
The following two conditions are equivalent:
\begin{enumerate}[(i)]
    \item There exists $\bv \in \R^d \setminus \{0^d\}$, such that $\langle A_1, \ldots, A_k \rangle \cdot \bv \subset \mC$.
    \item The orbit $\langle A_1^{\top}, \ldots, A_k^{\top} \rangle \cdot \mC^*$ is contained in a closed halfspace.
\end{enumerate}
\end{restatable}

The proof of Lemma~\ref{lem:dual} is given in Section~\ref{sec:lintopos}.
Deciding whether an orbit has the property of being contained in a closed halfspace is challenging.
In particular, the smallest cone containing the orbit might not be finitely generated.
It is very unlikely that one can compute any precise description of this cone, otherwise one would be able to decide whether it is contained in a \emph{given} halfspace, and hence solve the Positivity Problem, entailing breakthroughs in number theory~\cite{ouaknine2014positivity}.
Therefore, we will first devise a procedure that decides a similar property for the orbit, namely whether it is salient.
We will later use it as a sub-procedure to decide whether the orbit is contained in a closed halfspace.
The next lemma reduces deciding salient-ness to finding certain tuples of polynomials in $\A^+$.

\begin{restatable}{lem}{lemequiv}\label{lem:equiv}
The following two conditions are equivalent:
\begin{enumerate}[(i)]
    \item The orbit $\langle A_1^{\top}, \ldots, A_k^{\top} \rangle \cdot \mC^*$ is not salient.
    \item There exists a tuple of polynomials $(f_1, \ldots, f_n) \in \left(\A^+\right)^n \setminus \{0^n\}$, not all zero, such that $\sum_{i=1}^n f_i(A_1^{\top}, \ldots, A_k^{\top}) \cdot \bc_i = 0^d$.
\end{enumerate}
\end{restatable}

The proof of Lemma~\ref{lem:equiv} is given in Section~\ref{sec:lintopos}.

Define the following $\A$-submodule of $\A^n$.
\begin{equation}\label{eq:defM}
    \mM \coloneqq \left\{(f_1, \ldots, f_n) \in \A^n \;\middle|\; \sum_{i=1}^n f_i(A_1^{\top}, \ldots, A_k^{\top}) \cdot \bc_i = \bzer \right\}.
\end{equation}
One easily verifies that $\bg,\bh \in \mM \implies \bg+\bh \in \mM$, as well as $f \in \A, \bg \in \mM \implies f \cdot \bg \in \mM$; so $\mM$ is indeed an $\A$-module.
Observe that condition (ii) of Lemma~\ref{lem:equiv} can be expressed as ``$\mM$ contains an element of $\left(\A^+\right)^n \setminus \{0^d\}^n$''.
The next proposition shows that a basis of $\mM$ is computable.

\begin{restatable}{prop}{propM}\label{prop:M}
Given as input $A_1, \ldots, A_k$ and $\bc_1, \ldots, \bc_n$, one can compute a basis $\bg_1, \ldots, \bg_m \in \A^n$ such that $\mM = \bg_1 \A + \cdots + \bg_m \A$.
\end{restatable}

The proof of Proposition~\ref{prop:M} is given in Section~\ref{sec:M}.
The next proposition shows that it is decidable whether the module $\mM$ contains an element of $\left(\A^+\right)^n \setminus \{0^n\}$.
Hence condition (ii) of Lemma~\ref{lem:equiv} is decidable.

\begin{restatable}{prop}{proppos}\label{prop:pos}
There is a procedure that, given as input a basis $\bg_1, \ldots, \bg_m \in \A^n$ with rational coefficients such that $\mM = \bg_1 \A + \cdots + \bg_m \A$, decides whether $\mM$ contains an element of $\left(\A^+\right)^n \setminus \{0^n\}$,
and outputs this element $\bff$ in case it exists.
\end{restatable}

The proof of Proposition~\ref{prop:pos} is given in Section~\ref{sec:pos}.
Combining Lemma~\ref{lem:equiv}, Proposition~\ref{prop:M} and \ref{prop:pos}, we conclude that it is decidable whether the orbit $\langle A_1^{\top}, \ldots, A_k^{\top} \rangle \cdot \mC^*$ is salient.
We now use this to show that it is decidable whether $\langle A_1^{\top}, \ldots, A_k^{\top} \rangle \cdot \mC^*$ is contained in a closed halfspace.
This together with Lemma~\ref{lem:dual} will give us the main result:

\begin{restatable}{thrm}{thrmmain}\label{thm:main}
Given commuting invertible matrices $A_1, \ldots, A_k \in \Q^{d \times d}$ and vectors $\bc_1, \ldots, \bc_n \in \Q^d$ defining a polyhedral cone $\mC$ as in Equation~\eqref{eq:defC}, it is decidable whether there exists $\bv \in \R^d \setminus \{0^d\}$ such that $\langle A_1, \ldots, A_k \rangle \cdot \bv \subset \mC$.
\end{restatable}
\begin{proof}
By Lemma~\ref{lem:dual}, it suffices to construct a procedure that decides whether the orbit $\langle A_1^{\top}, \ldots, A_k^{\top} \rangle \cdot \mC^*$ is contained in a closed halfspace.
We use induction on $d$.
When $d = 1$, the decision procedure is easy: $\langle A_1^{\top}, \ldots, A_k^{\top} \rangle \cdot \mC^*$ is contained in a closed halfspace if and only if $\bc_1, \ldots, \bc_n$ all have the same sign and $A_1, \ldots, A_k$ are all positive rationals.

Suppose we have a decision procedure for all dimensions smaller than $d$, we now construct a procedure for dimension $d$.
By Lemma~\ref{lem:equiv}, Proposition~\ref{prop:M} and \ref{prop:pos}, it is decidable whether the orbit $\langle A_1^{\top}, \ldots, A_k^{\top} \rangle \cdot \mC^*$ is salient.

If the orbit is salient, then it is contained in some closed halfspace.
If the orbit is not salient, then $\cone\left(\langle A_1^{\top}, \ldots, A_k^{\top} \rangle \cdot \mC^*\right)$ contains the vectors $\bw$ and $- \bw$ for some $\bw \in \R^d \setminus \{0^d\}$.
Such a $\bw$ can be effectively computed in the following way:
by Proposition~\ref{prop:pos}, one can compute an element $\bff \in \mM \cap \left(\A^+\right)^n \setminus \{0^n\}$.
Then, as in the proof of Lemma~\ref{lem:equiv} (see Section~\ref{sec:lintopos}), one can obtain from $\bff$ a non-zero vector $\bw \in \Q^d$ such that $\pm \bw \in \cone\left(\langle A_1^{\top}, \ldots, A_k^{\top} \rangle \cdot \mC^*\right)$.

We then compute (a basis of) the linear space 
\[
W \coloneqq \lin\left(\langle A_1^{\top}, \ldots, A_k^{\top} \rangle \cdot \bw\right).
\]
This can be done by starting with $W \coloneqq \R \bw$ and repeatedly replacing $W$ by $W + A_1 W + \cdots + A_k W$ until $W = W + A_1 W + \cdots + A_k W$.
This process must terminate since the dimension of $W$ is always growing.
The resulting $W = \lin\left(\langle A_1^{\top}, \ldots, A_k^{\top} \rangle \bw\right)$ is invariant under the linear maps $A_1^{\top}, \ldots, A_k^{\top}$.
Since $\R \bw \subset \cone\left(\langle A_1^{\top}, \ldots, A_k^{\top} \rangle \cdot \mC^*\right)$ we have $W \subset \cone\left(\langle A_1^{\top}, \ldots, A_k^{\top} \rangle \cdot \mC^*\right)$.

The quotient $\R^d/W$ is isomorphic to $\R^{d - \dim(W)}$ after fixing a basis.
We can suppose this basis admits only rational entries.
Composing with the canonical map $\R^d \rightarrow \R^d/W$ we have an (effectively computable) map
\[
\pi \colon \R^d \longrightarrow \R^{d - \dim(W)}
\]
with $\ker(\pi) = W$.
Since multiplication by the matrices $A_1^{\top}, \ldots, A_k^{\top}$ leaves $W$ invariant, these matrices act on $\R^d/W \cong \R^{d - \dim(W)}$ linearly.
We denote these actions by 
\[
B_1, \ldots, B_k \in \Q^{(d - \dim(W)) \times (d - \dim(W))}.
\]
These matrices are obviously invertible, commuting and effectively computable.

Using the induction hypothesis on the dimension $d$, we can decide whether the image
\[
\pi\left(\langle A_1^{\top}, \ldots, A_k^{\top} \rangle \cdot \mC^* \right) = \langle B_1, \ldots, B_k \rangle \cdot \pi(\mC^*)
\]
is contained in some closed halfspace (of $\R^d/W \cong \R^{d - \dim(W)}$).
If the image $\pi\left(\langle A_1^{\top}, \ldots, A_k^{\top} \rangle \cdot \mC^* \right)$ is contained in a closed halfspace $\mH \subset \R^{d - \dim(W)}$ then $\langle A_1^{\top}, \ldots, A_k^{\top} \rangle \cdot \mC^*$ is contained in the closed halfspace $\pi^{-1}(\mH) \subset \R^d$.
Otherwise, 
\[
\cone\left(\pi\left(\langle A_1^{\top}, \ldots, A_k^{\top} \rangle \cdot \mC^* \right)\right) = \R^{d - \dim(W)}.
\]
Since $\pi(\bv + \bw) = \pi(\bv) + \pi(\bw)$ for all $\bv, \bw \in \R^d$, the above equation yields
\[
\pi\left(\cone\left(\langle A_1^{\top}, \ldots, A_k^{\top} \rangle \cdot \mC^* \right)\right) = \R^{d - \dim(W)}.
\]
Thus, any $\bv \in \R^d$ can be written as $\bv' + \bv''$ where 
\[
\bv' \in \cone\left(\langle A_1^{\top}, \ldots, A_k^{\top} \rangle \cdot \mC^* \right),
\]
and $\bv'' \in W$.
Recall that $W \subset \cone\left(\langle A_1^{\top}, \ldots, A_k^{\top} \rangle \cdot \mC^*\right)$.
Therefore, $\bv = \bv' + \bv'' \in \cone\left(\langle A_1^{\top}, \ldots, A_k^{\top} \rangle \cdot \mC^*\right)$ for all $\bv \in \R^d$.
So 
\[
\cone\left(\langle A_1^{\top}, \ldots, A_k^{\top} \rangle \cdot \mC^*\right) = \R^d
\]
is not contained in any closed halfspace.
\end{proof}

\section{From linear loops to positive polynomials}\label{sec:lintopos}

In this section we give the proofs of Lemma~\ref{lem:dual} and \ref{lem:equiv}.

\lemdual*
\begin{proof}
    (i) $\implies$ (ii).
    Suppose $\bv \in \R^d \setminus \{0^d\}$, such that $\langle A_1, \ldots, A_k \rangle \cdot \bv \subset \mC$.
    Then by the definition of $\mC$, for all $n_1, \ldots, n_k \in \N$ and $i = 1, \ldots, n$, we have $\bc_i^{\top} A_1^{n_1} \cdots A_k^{n_k} \bv \geq 0$.
    Taking the transpose yields
    \begin{equation}\label{eq:transpose}
        \bv^{\top} \left(A_k^{\top}\right)^{n_k} \cdots \left(A_1^{\top}\right)^{n_1} \bc_i \geq 0, \quad i = 1, \ldots, n; n_1, \ldots, n_k \in \N.
    \end{equation}
    
    Let $\mH$ be the closed halfspace defined by $\mH \coloneqq \{\bx \in \R^d \mid \bv^{\top} \bx \geq 0\}$.
    Then Equation~\eqref{eq:transpose} shows that $\langle A_1^{\top}, \ldots, A_k^{\top} \rangle \cdot \bc_i \subset \mH$ for all $i$.
    Since the cone $\mC^*$ is generated by $\bc_1, \ldots, \bc_n$, we have $\langle A_1^{\top}, \ldots, A_k^{\top} \rangle \cdot \mC^* \subset \mH$.

    (ii) $\implies$ (i).
    Suppose $\langle A_1^{\top}, \ldots, A_k^{\top} \rangle \cdot \mC^*$ is contained in the halfspace $\mH$ defined by $\mH \coloneqq \{\bx \in \R^d \mid \bv^{\top} \bx \geq 0\}$, where $\bv \neq 0^d$.
    Then Equation~\eqref{eq:transpose} holds for all $i$ and all $n_1, \ldots, n_k \in \N$.
    Taking the transpose yields $\bc_i^{\top} A_1^{n_1} \cdots A_k^{n_k} \bv \geq 0$ and hence $\langle A_1, \ldots, A_k \rangle \cdot \bv \subset \mC$.
\end{proof}

\lemequiv*
\begin{proof}
    (i) $\implies$ (ii).
    If $\langle A_1^{\top}, \ldots, A_k^{\top} \rangle \cdot \mC^*$ is not salient, then there exists $\bv \neq 0^d$ with $\bv, -\bv \in \cone\left(\langle A_1^{\top}, \ldots, A_k^{\top} \rangle \cdot \mC^*\right)$.
    Therefore, there exist finitely many positive reals $r_{i, n_1, \ldots, n_k}$, such that
    \[
        \sum_{i = 1}^n \sum_{n_1, \ldots, n_k \in \N} r_{i, n_1, \ldots, n_k} \left(A_1^{\top}\right)^{n_1} \cdots \left(A_k^{\top}\right)^{n_k} \cdot \bc_i = \bv.
    \]
    
    In order words, there exist polynomials with non-negative coefficients $(g_1, \ldots, g_n) \in \left(\A^+\right)^n$, such that $\sum_{i=1}^n g_i(A_1^{\top}, \ldots, A_k^{\top}) \cdot \bc_i = \bv$.
    Similarly, there exist polynomials $(h_1, \ldots, h_n) \in \left(\A^+\right)^n$, such that $\sum_{i=1}^n h_i(A_1^{\top}, \ldots, A_k^{\top}) \cdot \bc_i = - \bv$.
    Since $\bv \neq 0^d$, we know that $g_1, \ldots, g_n$ are not all zero.
    Let $f_i \coloneqq g_i + h_i, i = 1, \ldots, n$, then $\sum_{i=1}^n f_i(A_1^{\top}, \ldots, A_k^{\top}) \cdot \bc_i = 0^d$; and $(f_1, \ldots, f_n) \in \left(\A^+\right)^n \setminus \{0^n\}$.

    (ii) $\implies$ (i).
    Suppose $\sum_{i=1}^n f_i(A_1^{\top}, \ldots, A_k^{\top}) \cdot \bc_i = 0^d$ for some $(f_1, \ldots, f_n) \in \left(\A^+\right)^n \setminus \{0^n\}$.
    Without loss of generality suppose $f_1 \neq 0$. Let $g$ be a monomial of $f_1$ and write $f_1 = g + f'_1$.
    Since $f_1 \in \A^+$, we have also $g, f'_1 \in \A^+$.
    We have $\bv \coloneqq g(A_1^{\top}, \ldots, A_k^{\top}) \cdot \bc_1 \neq 0^d$ because $A_1^{\top}, \ldots, A_k^{\top}$ are invertible, $\bc_1 \neq 0^d$ and because $g$ is a monomial.
    Then $\bv = g(A_1^{\top}, \ldots, A_k^{\top}) \cdot \bc_1 \in \langle A_1^{\top}, \ldots, A_k^{\top} \rangle \cdot \mC^*$ and
    \begin{multline*}
        - \bv = f'_1(A_1^{\top}, \ldots, A_k^{\top}) \cdot \bc_1 + \sum_{i=2}^n f_i(A_1^{\top}, \ldots, A_k^{\top}) \cdot \bc_i \\
        \in \cone\left(\langle A_1^{\top}, \ldots, A_k^{\top} \rangle \cdot \mC^*\right).
    \end{multline*}
    Thus, $\cone\left(\langle A_1^{\top}, \ldots, A_k^{\top} \rangle \cdot \mC^*\right)$ contains both $\bv$ and $-\bv$ and is therefore not salient.
\end{proof}

\section{Computing the module $\mM$}\label{sec:M}
In this section we prove Proposition~\ref{prop:M}.
Let $\be_1, \ldots, \be_n$ be the standard basis of $\A^n$,
that is, 
\[
\be_i = (0, \cdots, 0, 1, 0, \cdots, 0),
\]
where the $i$-th coordinate is one.

The vector space $\R^d$ can be considered as an $\A$-module by $f \cdot \bv \coloneqq f(A_1^{\top}, \ldots, A_k^{\top}) \bv$ for all $f \in \A, \bv \in \R^d$. 
Define the following map of $\A$-modules:
\begin{align}
    \varphi : \A^n & \longrightarrow \R^d, \\
            \sum_{i = 1}^n f_i \be_i & \mapsto \sum_{i = 1}^n f_i(A_1^{\top}, \ldots, A_k^{\top}) \bc_i.
\end{align}
Then the $\A$-module $\mM$ defined in~\eqref{eq:defM} is exactly $\ker(\varphi)$.

\propM*
\begin{proof}

    Let $F_1 \in \R[X_1], \ldots, F_k \in \R[X_k]$ be the (monic) characteristic polynomials of $A_1^{\top}, \ldots, A_k^{\top}$, respectively.
    Then $F_j(A_j^{\top}) = 0^{d \times d}$, so obviously $F_j \be_i \in \mM$ for all $1 \leq j \leq k, 1 \leq i \leq n$.
    Define the (finite) set of monomials
    \[
    S \coloneqq \{X_1^{n_1} \cdots X_k^{n_k} \be_i \mid n_1 < \deg F_1, \ldots, n_k < \deg F_k, 1 \leq i \leq n\}.
    \]
    Consider the map
    \begin{align}
        \lambda : \R^{\card(S)} & \longrightarrow \R^d, \\
                (r_s)_{s \in S} & \mapsto \sum_{s \in S} r_s \varphi(s).
    \end{align}
    One can effectively compute a basis $(r_{1s})_{s \in S}, \ldots, (r_{ps})_{s \in S}$ of $\ker(\lambda)$.
    Define the set 
    \[
    \mR \coloneqq \left\{ \sum_{s \in S} r_{js} \cdot s \;\middle|\; j = 1, \ldots, p \right\}
    \]
    of elements in $\A^n$.
    Clearly $\mR \subset \ker(\varphi) = \mM$.
    Define additionally the following subset of $\mM$: 
    \[
    \mF \coloneqq \left\{F_j \be_i \;\middle|\; 1 \leq j \leq k, 1 \leq i \leq n\right\}.
    \]
    We claim that $\mR \cup \mF$ is a basis for $\mM$.

    To prove this claim, let $L$ be the $\A$-module generated by $\mR \cup \mF$. It is a submodule of $\mM$.
    We will show $L = \mM$.
    Let $\bff = \sum_{i = 1}^n f_i \be_i$ be an element of $\mM$.
    
    Every monomial $m \coloneqq X_1^{n_1} \cdots X_k^{n_k}$ of every polynomial $f_i, i = 1, \ldots, n,$ can be written as two polynomials $m=p'+p''$ such that $p'' \be_i\in L$ and every monomial of $p'\be_i$ is in the set $S$.
    Indeed, denote $J \coloneqq \{1\leq j\leq k \mid n_j > \deg F_j\}$.
    We write $J = \{i_1, \ldots, i_d\}$ for some $d \leq k$, and where $i_1, \ldots, i_d, i_{d+1}, \ldots, i_k$ is some permutation of $(1, \ldots, k)$.
    We divide $X_j^{n_j}$ by $F_j$ for every $j \in J$ and obtain
\[
X_j^{n_j}=P_jF_j+R_j
\] 
with $\deg R_j < \deg F_j$. 
Then
\[
m-(P_{i_1}F_{i_1})X_{i_2}^{n_{i_2}} \cdots X_{i_k}^{n_{i_k}}=R_{i_1}X_{i_2}^{n_{i_2}}\cdots X_{i_k}^{n_{i_k}},
\]
with $(P_{i_1}F_{i_1})X_{i_2}^{n_{i_2}}\cdots X_{i_k}^{n_{i_k}} \be_i\in L$, and
\[
R_{i_1}X_{i_2}^{n_{i_2}}\cdots\,X_{i_k}^{n_{i_k}} - R_{i_1}(P_{i_2}F_{i_2})X_{i_3}^{n_{i_3}}\cdots\,X_{i_k}^{n_{i_k}}=R_{i_1}R_{i_2}X_{i_3}^{n_{i_3}}\cdots X_{i_k}^{n_{i_k}},
\]
with $R_{i_1}(P_{i_2}F_{i_2})X_{i_3}^{n_{i_3}}\cdots\,X_{i_k}^{n_{i_k}} \be_i\in L$.
Continue this process until $i_d$, we obtain $m = p' + p''$ where
\[
p' = R_{i_1} \cdots R_{i_d} X_{i_{d+1}}^{n_{i_{d+1}}}\cdots X_{i_k}^{n_{i_k}},
\]
and 
\begin{multline*}
    p'' = (P_{i_1}F_{i_1})X_{i_2}^{n_{i_2}} \cdots X_{i_k}^{n_{i_k}} + R_{i_1}(P_{i_2}F_{i_2})X_{i_3}^{n_{i_3}}\cdots\,X_{i_k}^{n_{i_k}} + \\
    \cdots + R_{i_1} \cdots R_{i_{d-1}} (P_{i_d} F_{i_d}) X_{i_{d+1}}^{n_{i_{d+1}}}\cdots X_{i_k}^{n_{i_k}}.
\end{multline*}
Clearly, every monomials of $p'$ is in the set $\mS$, and $p'' \be_i\in L$.

Decompose every monomial of every $f_i, i = 1, \ldots, n$.
In this way, $\bff$ can be written as $\bff = \bff' + \bff''$ with $\bff'' \in L$ and $\bff'$ generated by the set $\mR$ (note that every monomial of $\bff'$ is in the set $\mS$ and $\bff' \in \mM$).

    We have thus found the finite basis $\mR \cup \mF$ for $\mM$ as a $\A$-module.
    Note that the procedure described in the proof is effective: the computation of the characteristic polynomials $F_1, \ldots, F_k$ can be done by computing the determinants $\det(A_i - X I), i = 1, \ldots, k$, and the computation of $\ker(\lambda)$ is simply linear algebra.

\end{proof}

Note that since the entries of the matrices $A_1, \ldots, A_k$ and the vectors $\bc_1, \ldots, \bc_n$ are all rationals, the computation given in Proposition~\ref{prop:M} can be done over the rational.
Hence, we can suppose the computed basis of $\mM$ has rational coefficients.

\section{Positive polynomial membership}\label{sec:pos}
In this section we prove Proposition~\ref{prop:pos}.
Recall $\A^{++} \coloneqq \A^+ \setminus \{0\}$.
Let $\B \coloneqq \R[X_1^{\pm 1}, \ldots, X_k^{\pm 1}]$ be the Laurent polynomial ring over $k$ variables, and $\B^+ \coloneqq \R_{\geq 0}[X_1^{\pm 1}, \ldots, X_k^{\pm 1}]$ be the sub-semiring of Laurent polynomials with non-negative coefficients.
Define $\B^{++} \coloneqq \B^+ \setminus \{0\}$.

The first step to proving Proposition~\ref{prop:pos} is a quick reduction from $\left(\A^{+}\right)^n \setminus \{0^n\}$ to $\left(\A^{++}\right)^n$.
\begin{lem}\label{lem:+to++}
    Given an $\A$-submodule $\mM$ of $\A^n$ and a non-empty subset $I \subset \{1, \ldots, n\}$, define the following $\A$-submodule of $\A^I$:
    \[
    \mM_{I} \coloneqq \left\{ (f_i)_{i \in I} \;\middle|\;\sum_{i \in I} f_i \be_i \in \mM\right\} \subset \A^I.
    \]
    Then $\mM$ contains an element in $\left(\A^{+}\right)^n \setminus \{0^n\}$ if and only if there exists some non-empty set $I \subset \{1, \ldots, n\}$ such that $\mM_{I} \cap \left(\A^{++}\right)^{I} \neq \emptyset$.
\end{lem}
\begin{proof}
    If $\mM$ contains an element $\bff = \sum_{i = 1}^n f_i \be_i \in \left(\A^{+}\right)^n \setminus \{0^n\}$.
    Let $I$ be the set $\{i \mid f_i \neq 0\}$.
    Then $(f_i)_{i \in I} \in \mM_I \cap \left(\A^{++}\right)^{I} \neq \emptyset$.

    If there exists non-empty $I \subset \{1, \ldots, n\}$ such that $\mM_{I} \cap \left(\A^{++}\right)^{I} \neq \emptyset$.
    Let $(f_i)_{i \in I} \in \mM_{I} \cap \left(\A^{++}\right)^{I}$, then $\sum_{i = 1}^n f_i \be_i \in \mM \cap \left(\A^{+}\right)^n \setminus \{0^n\} \neq \emptyset$.
\end{proof}
Given a finite basis for the $\A$-submodule $\mM$ of $\A^n$, it is easy to compute $\mM_I$ for any $I \subset \{1, \ldots, n\}$ using linear algebra over $\A$~\cite{bareiss1968sylvester}.
Therefore it suffices to devise a procedure that decide whether $\mM \cap \left(\A^{++} \right)^n \neq \emptyset$ for a given $\mM$.
This procedure can then be used to decide whether $\mM \cap \left(\A^{+}\right)^n \setminus \{0^n\} \neq \emptyset$ by verifying whether there exists non-empty $I \subset \{1, \ldots, n\}$ such that $\mM_I \cap \left(\A^{++}\right)^{I} \neq \emptyset$.

\begin{lem}\label{lem:AtoB}
    Let $\mM = \bg_1 \A + \cdots + \bg_m \A$ be an $\A$-submodule of $\A^n$.
    Define $\mM_{\B} \coloneqq \bg_1 \B + \cdots + \bg_m \B$; it is a $\B$-submodule of $\B^n$.
    Then $\mM \cap \left(\A^{++} \right)^n \neq \emptyset$ if and only if $\mM_{\B} \cap \left(\B^{++} \right)^n \neq \emptyset$.
\end{lem}
\begin{proof}
    Obviously $\mM \subset \mM_{\B}$ and $\A^{++} \subset \B^{++}$, so $\mM \cap \left(\A^{++} \right)^n \neq \emptyset \implies \mM_{\B} \cap \left(\B^{++} \right)^n \neq \emptyset$.
    
    For the other direction, suppose $\mM_{\B} \cap \left(\B^{++} \right)^n \neq \emptyset$.
    Let $\bff = \sum_{i = 1}^m h_i \bg_i \in \left(\B^{++} \right)^n$.
    Then there exists $n_1, \ldots, n_k \in \N$ such that $X_1^{n_1} \cdots X_k^{n_k} h_i \in \A$ for all $1 \leq i \leq m$.
    So 
    \[
    X_1^{n_1} \cdots X_k^{n_k} \bff = \sum_{i = 1}^m X_1^{n_1} \cdots X_k^{n_k} h_i \bg_i \in \mM \cap \left(\A^{++} \right)^n.
    \]
\end{proof}

By Lemma~\ref{lem:AtoB}, it suffices to consider finitely generated $\B$-submodules of $\B^n$.
The key ingredient of our proof is the following consequence of a deep result by Einsiedler, Mouat, and Tuncel~\cite{einsiedler2003does}.

\begin{lem}[{Corollary of \cite[Procedure~6.3]{einsiedler2003does}}]\label{lem:cor}
    Suppose a finite basis of a submodule $\mM_{\B}$ of $\B^n$ is given.
    One can decide whether $\mM_{\B} \cap \left(\B^{++} \right)^n \neq \emptyset$, provided an effective procedure for the following problem:

    \begin{itemize}
    \item \textbf{(ExistPos):} Given a rectangular set $K \coloneqq [p, q]^d \subset \R_{>0}^d$, decide if for all $\ba \in K$ there exists $\bff = (f_1, \dots, f_n) \in \mM_{\B}$ such that $\bff(\ba) \coloneqq \left(f_1(\ba), \ldots, f_n(\ba)\right) \in \R_{>0}^n$.
    \end{itemize}

    Furthermore, we can find an element of $\mM_{\B} \cap \left(\B^{++} \right)^n$ with rational coefficients if the set is non-empty.
\end{lem}
\begin{proof}
    Consider the result of Einsiedler, Mouat, and Tuncel~\cite[Procedure~6.3]{einsiedler2003does}.
    By the remark following Procedure~6.3, this procedure becomes effective if \cite[Condition~1.3(a)]{einsiedler2003does} can be checked algorithmically for the compact set $K$ defined in \cite[Lemma~5.3]{einsiedler2003does}.\footnote{The second question mentioned in the remark (whether there is a computable bound on $\ell$ in the procedure?) does not hinder decidability since the existence of $\ell$ is guaranteed, according to the remark.}
    Therefore, it suffices to provide an effective procedure that checks \cite[Condition~1.3(a)]{einsiedler2003does} for the compact set $K$.
    By the proof of \cite[Lemma~5.3]{einsiedler2003does}, one can even suppose $K$ to be of the form $[p, q]^d \subset \R_{>0}^d$.
    Then \cite[Condition~1.3(a)]{einsiedler2003does} is equivalent to the problem \textbf{ExistPos}.
\end{proof}


\begin{lem}\label{lem:dual2}
    Fix an $\ba \in \R_{> 0}^k$ and let $\mM_{\B} = \bg_1 \B + \cdots + \bg_m \B$.
    Then the two following conditions are equivalent:
    \begin{enumerate}[(i)]
        \item There exists $\bff = (f_1, \dots, f_n) \in \mM_{\B}$ such that $\bff(\ba) \in \R_{>0}^n$.
        \item There exists reals $r_1, \ldots, r_m \in \R$ such that $\sum_{i = 1}^m r_i \bg_i(\ba) \in \R_{>0}^n$.
    \end{enumerate}
\end{lem}
\begin{proof}
    (i) $\implies$ (ii).
    Let $\bff \in \mM_{\B}$ such that $\bff(\ba) \in \R_{>0}^n$.
    Suppose $\bff = \sum_{j = 1}^m h_j \bg_j$ for some $h_1, \ldots, h_m \in \B$, then we have $\bff(\ba) = \sum_{j = 1}^m h_j(\ba) \bg_j(\ba)$.
    Taking $r_1 \coloneqq h_1(\ba), \ldots, r_m \coloneqq h_m(\ba)$ yields (ii).

    (ii) $\implies$ (i).
    Let $r_1, \ldots, r_m \in \R$ be reals such that $\sum_{i = 1}^m r_i \bg_i(\ba) \in \R_{>0}^n$.
    Taking $\bff \coloneqq \sum_{i = 1}^m r_i \bg_i \in \mM$ yields (i).
\end{proof}

\begin{prop}\label{prop:dec++}
    Suppose a finite basis of a submodule $\mM$ of $\A^n$ is given.
    One can decide whether $\mM \cap \left(\A^{++} \right)^n \neq \emptyset$.
\end{prop}
\begin{proof}
    By Lemma~\ref{lem:AtoB}, it suffices to decide whether $\mM_{\B} \cap \left(\B^{++}\right)^n \neq \emptyset$.
    By Lemma~\ref{lem:cor}, it suffices to give a decision procedure for the problem \textbf{ExistPos}.
    By Lemma~\ref{lem:dual2}, \textbf{ExistPos} has a positive answer if and only if the following statement in the first order theory of reals is true:
    \begin{align}
    & \forall a_1 \cdots \forall a_d \exists r_1 \cdots \exists r_m, \nonumber \\
    & \bigg(a_1 \in [p,q] \land \cdots \land a_d \in [p,q]\bigg)
    \implies \sum_{i = 1}^m r_i \bg_i(a_1, \ldots, a_d) \in \R_{>0}^n.
    \end{align}
    The truth of such statements is decidable by Tarski's theorem~\cite{Tarski1949}.
\end{proof}

Combining Lemma~\ref{lem:+to++} and Proposition~\ref{prop:dec++} immediately yields Proposition~\ref{prop:pos}.

\proppos*
\begin{proof}
    By Lemma~\ref{lem:+to++}, $\mM$ contains an element of $\left(\A^+\right)^n \setminus \{0^n\}$ if and only if there exists a non-empty set $I \subset \{1, \ldots, n\}$ such that $\mM_I \cap \left(\A^{++}\right)^{I} \neq \emptyset$.
    Therefore, it suffices to enumerate all non-empty sets $I \subset \{1, \ldots, n\}$, compute a basis for $\mM_I$, and use Proposition~\ref{prop:dec++} (with $n \coloneqq \card(I)$) to check whether $\mM_I \cap \left(\A^{++}\right)^{I} \neq \emptyset$.
\end{proof}

\section{Conclusion and future work}
In this paper we proved decidability of termination of linear loops under commutative updates, where the guard condition is a polyhedral cone.
A natural continuation of our work would be to generalize Theorem~\ref{thm:main} to the case where $\mC$ is a polyhedron (instead of a polyhedral cone).
Unfortunately, the duality argument that is crucial in our proof stops working when we remove the homogeneity of the guard conditions.
New techniques, possibly combining geometric and algebraic arguments, might be needed to overcome this difficulty.
Another possible generalization is the removal of the commutativity assumption on the update matrices.
For this, one would need a generalization of Einsiedler, Mouat, and Tuncel's result to left-modules over non-commutative polynomials rings, which would have deep consequences in the field of real algebra.

\bibliography{pointed}

\end{document}